\documentclass[journal]{IEEEtran}
\ifCLASSINFOpdf
\else
\fi

\usepackage{graphicx}
\usepackage{subfigure}
\usepackage{amsthm}
\usepackage{times} % assumes new font selection scheme installed
\usepackage{amsmath} % assumes amsmath package installed
\usepackage{amssymb}  % assumes amsmath package installed
\usepackage{cleveref}
\usepackage{url} 
\usepackage{color,xcolor,soul}
\usepackage{cite}
\usepackage{url} 
\usepackage{amsfonts}
\usepackage{comment} 
\usepackage{enumitem}
\setlist[enumerate]{itemsep=0mm, parsep=0mm}

\makeatletter
\newcommand{\pushright}[1]{\ifmeasuring@#1\else\omit\hfill$\displaystyle#1$\fi\ignorespaces}
\newcommand{\pushleft}[1]{\ifmeasuring@#1\else\omit\displaystyle#1\hfill\fi\ignorespaces}
\newcommand{\specialcell}[1]{\ifmeasuring@#1\else\omit$\displaystyle#1$\ignorespaces\fi}
\makeatother
\definecolor{darkGreen}{rgb}{0,0.5,0.1}

\newtheorem{theorem}{Theorem}
\newtheorem{proposition}{Proposition}
\newtheorem{remark}{Remark}
\newtheorem{definition}{Definition}
\newtheorem{assumption}{Assumption}

\usepackage{ifthen}
\newboolean{showcomments}
\setboolean{showcomments}{true}

\newcommand{\highlight}[1]{\ifthenelse{\boolean{showcomments}} {\textcolor{red}{#1}}{}}
\newcommand{\edited}[1]{\ifthenelse{\boolean{showcomments}} {\textcolor{blue}{#1}}{}}
\newcommand{\jamie}[1]{\ifthenelse{\boolean{showcomments}}
{ \textcolor{red}{(Jamie says:  #1)}}{}} 
\newcommand{\sai}[1]{\ifthenelse{\boolean{showcomments}}
{ \textcolor{red}{(Sai says:  #1)}}{}}
\newcommand{\soumya}[1]{\ifthenelse{\boolean{showcomments}}
{ \textcolor{red}{(Soumya says:  #1)}}{}}
\begin{document}

% OPTION 1
% \title{Robust Demand-Side Control for Frequency Response in the Presence of Uncertain Loads}

% OPTION 2
% \title{Primary Frequency Control of Uncertain Responsive Loads with Stochastic Stability Guarantees}

% OPTION 3
\title{Sparse Control Synthesis for Uncertain Responsive Loads with Stochastic Stability Guarantees}


\author{Sai~Pushpak~Nandanoori,~\IEEEmembership{Member,~IEEE,}
        Soumya~Kundu,~\IEEEmembership{Member,~IEEE,}
        Jianming~Lian,~\IEEEmembership{Member,~IEEE,}
        Umesh~Vaidya,~\IEEEmembership{Senior Member,~IEEE,} 
        Draguna~Vrabie,~\IEEEmembership{Member,~IEEE,}
        and~Karanjit Kalsi,~\IEEEmembership{Senior Member,~IEEE}% <-this % stops a space
\thanks{This work was partially supported by the ARPA-E NODES program of the U.S. Department of Energy (DOE).} 
\thanks{S. P. Nandanoori, S. Kundu, D. Vrabie and K. Kalsi are with Pacific Northwest National Laboratory, Richland, WA, 99354, USA, which is operated by Battelle under Contract No. DE-AC05-76RL01830 with the U.S. DOE. (email: \{saipushpak.n, soumya.kundu, draguna.vrabie, karanjit.kalsi\}@pnnl.gov)}
\thanks{J. Lian is with Oak Ridge National Laboratory, Oak Ridge, TN, 37831, USA, which is operated by UT-Battelle, LLC under Contract No. DE-AC05-00OR22725 with the U.S. DOE. (email: lianj@ornl.gov)}
\thanks{U. Vaidya is with Clemson University, Clemson, SC, 29634, USA. (email: uvaidya@clemson.edu)}% <-this % stops a space
% \thanks{Manuscript received April 19, 2005; revised August 26, 2015.}
}

\maketitle

\begin{abstract}
Recent studies have demonstrated the potential of flexible loads in providing frequency response services. However, uncertainty and variability in various weather-related and end-use behavioral factors often affect the demand-side control performance. This work addresses this problem with the design of a demand-side control to achieve frequency response under load uncertainties. Our approach involves modeling the load uncertainties via stochastic processes that appear as both multiplicative and additive to the system states in closed-loop power system dynamics. 
Extending the recently developed mean square exponential stability (MSES) results for stochastic systems, we formulate multi-objective linear matrix inequality (LMI)-based optimal control synthesis problems to not only guarantee stochastic stability, but also promote sparsity, enhance closed-loop transient performance, and maximize allowable uncertainties. 
% 
% Recently developed mean square exponential stability (MSES) results for continuous-time linear stochastic systems are applied for the controller synthesis resulting in an linear matrix inequality (LMI)-based optimization problem. In particular, constrained multi-objective LMI problems are formulated to not only guarantee MSES, but also improve closed-loop transient performance, maximize tolerable uncertainties, and promote sparsity in the controller. 
% 
The fundamental trade-off between the maximum allowable (\textit{critical}) uncertainty levels and the optimal stochastic stabilizing control efforts is established. Moreover, the sparse control synthesis problem is generalized to the realistic power systems scenario in which only partial-state measurements are available. Detailed numerical studies are carried out on IEEE 39-bus system to demonstrate the closed-loop stochastic stabilizing performance of the sparse controllers in enhancing frequency response under load uncertainties; as well as illustrate the fundamental trade-off between the allowable uncertainties and optimal control efforts.
% The proposed control synthesis is illustrated on an IEEE 39 bus system with rigorous studies to demonstrate the role of the weights associated with sparsity, closed-loop transient performance, allowable uncertainties, and control efforts while ensuring MSES and achieving frequency response.
\end{abstract}


% \fi
\IEEEpeerreviewmaketitle

\section{Introduction}
\IEEEPARstart{S}{ignificant} technological advances in recent years have paved the way for flexible electrical loads to participate in grid ancillary services, including improvement of transient performance via \textit{primary} and \textit{secondary} frequency response (see, for example, \cite{vrakopoulou2017chance,zhang2016distributionally,geng2020chance,callaway2011achieving,lian2016hierarchical,moya2014hierarchical,kundu2017assessment,nandanoori2018prioritized,bhattacharya2019incentive,molina2011decentralized,mallada2017optimal,angeli2011stochastic,nazir2017performance,voice2017stochastic,chertkov2018ensemble,mavalizadeh2020decentralized,li2015market,zhao2012fast,pushpak2017fragility} and the references therein). This type of demand-side control offers relatively faster,
cleaner and cost-effective alternatives to more traditional
stabilizing control measures (e.g. spinning reserves) deployed in the power grid, and can help mitigate the challenges associated with increasing penetration of renewable generation. In comparison to traditional generator-based reserves, however, load-based reserves can be highly uncertain and hence unreliable in providing the required ancillary services \cite{vrakopoulou2017chance}. Engaging flexible and responsive loads for critical grid support, therefore, requires demand-side control synthesis methods that can structurally account for those uncertainties and minimize their impact on the grid transient performance. In the often-used hierarchical load control architectures \cite{callaway2011achieving,moya2014hierarchical,bhattacharya2019incentive,nandanoori2018prioritized,kundu2017assessment}, mitigation of uncertainties may happen at multiple layers, namely, the load aggregation and control layer, as well as the system-level grid operation layer. Several recent efforts explored methods to characterize and mitigate uncertainties in the controllable loads at the aggregator level, including probabilistic assessment of flexibility in providing grid ancillary services at multiple time-scales \cite{kundu2017assessment,chakraborty2019modeling,chakraborty2020stochastic}, priority-based coordination schemes to ensure reliability and mitigate uncertainty of the grid services \cite{nandanoori2018prioritized,mavalizadeh2020decentralized}, joint parameter-state estimation approaches \cite{li2015market}, as well as various Markov-based control schemes and randomized switching strategies for thermostatic loads \cite{angeli2011stochastic,nazir2017performance,voice2017stochastic,chertkov2018ensemble}. In yet other efforts, the grid-operator perspective of managing load uncertainties at the system-level have been considered primarily via various formulations of chance-constrained optimization techniques for capacity reserve determination and optimal dispatch of flexible loads \cite{zhang2016distributionally,vrakopoulou2017chance,geng2020chance}.

In comparison, the problem of quantifying and mitigating the impact of uncertain controllable loads on the system-wide transient performance of the bulk power grid has received relatively less attention. Fragility of the closed-loop power systems with load-frequency control to stochastic parameter uncertainties (such as renewable generation) was investigated in \cite{pushpak2017fragility}. The work \cite{zhao2012fast} studies the effect of stochastic measurement errors on the demand-side frequency control. Robustness of the demand-side frequency control performance to bounded disturbances in the closed-loop power system dynamics, arising due to unknown (model) parameters and unmodeled nonlinearity, have been considered in \cite{moya2014hierarchical,mallada2017optimal}. The existing methods have mostly focused on designing robust demand-side control strategies that can guarantee acceptable transient performance in presence of model uncertainties which are typically modeled as bounded additive disturbances. A key missing piece, however, is the treatment of stochastic uncertainties that arise in the demand-side feedback control action itself, and appear as multiplicative to the system states. This work addresses this issue by designing a sparse demand-side controller that can guarantee closed-loop stochastic stability performance in presence of multiplicative (and additive) uncertainties.

In particular, \cite{lian2016hierarchical,moya2014hierarchical,kundu2017assessment,nandanoori2018prioritized,bhattacharya2019incentive,molina2011decentralized} {considers the distributed and hierarchical demand-side primary frequency control methods.} In these methods, the participating thermostatic loads are coordinated in a distributed fashion to extract stabilizing response to disturbances in the grid, often in the form of proportional feedback control such as the power-frequency droop response. Uncertainties in load behavior, however, show up as errors in the controller gains, e.g. discrepancies in the actual vs. desired slope of the droop response \cite{kundu2017assessment,nandanoori2018prioritized}. These uncertainties can be modeled as multiplicative to the system states (e.g. frequencies) having a nonlinear impact on the system dynamics, unlike the additive uncertainties. {Mean square stability analysis of a stochastic power network where the uncertainty in system inertia appears multiplicative to the system dynamics is studied in} \cite{guo2019performance}. 
Recent efforts have looked into the problem of control synthesis with closed-loop stochastic stability guarantees in dynamical systems with both multiplicative and additive uncertainties \cite{nandanoori2018mean,elia2005remote,vaidya2010limitations,diwadkar2014stabilization}. 
% Recently developed MSES results for continuous-time network systems are applied to synthesize a stabilizing controller. 
% 
% 
Specifically, the work in \cite{nandanoori2018mean} provides sufficient and necessary conditions for mean square exponentially stability (MSES) for stochastic dynamical systems with multiplicative uncertainties. In addition, a static full-state feedback control synthesis was proposed in \cite{nandanoori2018mean} for the special case of single-input channel uncertainty, with a quantification on the maximum allowable (critical) uncertainty level. The approach in \cite{nandanoori2018mean}, however, has limited applicability in realistic power systems because of the assumptions placed on the nature of the uncertainty (single-input) and the availability of the measurements (full-state); and the non-consideration of design concerns such as sparsity, control efforts and closed-loop transient performance.

{The main contributions can be summarized as follows:}
\begin{enumerate}
    \item[a)] {modeling of multi-variate uncertainties (both multiplicative and additive) in the  controllable and uncontrollable loads resulting in a stochastic power network;} 
    \item[b)] {an LMI-based control synthesis problem for the stochastic system to guarantee closed-loop stochastic stability under multi-variate uncertainties along with a quantification of the allowable uncertainty levels;}
    \item[c)] {modification to the LMI control synthesis problem to address practical considerations in real power networks such as partial observation, desired transient performance, and sparsity promoting control; and}
    \item[d)] {numerical studies performed on IEEE 39-bus system to demonstrate the fundamental trade-offs between the allowable uncertainties and the optimal control efforts, and to illustrate the advantages of the proposed sparse controller over a conventional linear-quadratic-regulator (LQR) controller during contingency scenarios (e.g., severe under- and over-frequency events).}
\end{enumerate}

The rest of the paper is organized as follows. {Sec.}\,\ref{Sec_model} {describes the stochastic power system model with controllable and uncontrollable loads. Sec.} \ref{Sec_problem} {presents the stochastic stability definitions and formally introduces the stochastic stabilizing control design problem.} The convex optimization formulations for the sparse control synthesis with stochastic stability guarantees is presented in Sec.\,\ref{sec_robust_control}. Sec.\,\ref{Sec_simulation} illustrates the closed-loop controller performance via numerical studies on the IEEE 39-bus system; before we end the paper with concluding remarks in Sec.\,\ref{Sec_conclusion}.

\section{Stochastic Power System Model}
\label{Sec_model} 

We consider a power systems network with controllable loads which are coordinated to provide distributed frequency response for enhanced transient stability performance \cite{lian2016hierarchical,moya2014hierarchical,kundu2017assessment,nandanoori2018prioritized,bhattacharya2019incentive,molina2011decentralized}. Close to the nominal operating point, the power system dynamics can be modeled in the following control-affine form:
% \begin{subequations}\label{E:sys}
% \begin{align}
% \dot x &= A_0\, x + B_0 P_{in} \\
% y &= C_0\, x
% \end{align}
% \end{subequations}
\begin{align}\label{E:sys}
\dot x &= A_0\, x + B_0 P_{in}\,,\quad y = C_0\, x\,,
\end{align}
where $x\!\in\!\mathbb{R}^n$ is an $n$-dimensional state vector; $P_{in}\!\in\!\mathbb{R}^m$ is the $m$-dimensional vector of power injection inputs at various buses across the network; $y\!\in\!\mathbb{R}^s$ is the $s$-dimensional sensor measurements vector; $A_0\!\in\!\mathbb{R}^{n\times n},B_0\!\in\!\mathbb{R}^{n\times m}$ and $C_0\!\in\!\mathbb{R}^{s\times n}$ are, respectively, the system matrix, input matrix and output matrix. While the generic form of the system model shown in \eqref{E:sys} is sufficient for the development of the key technical arguments in this paper, a specific example of it is provided in Sec.\,\ref{Sec_simulation} for illustration purposes. We assume the matrices $A_0,B_0$ and $C_0$ to be deterministic, while the injected power inputs ($P_{in}$) are stochastic in nature. The following, typically non-restrictive, assumptions are made to allow the development of the technical results of this paper:
\begin{assumption}\label{AS:sys}

{$A_0$ is Hurwitz; and the pairs $(A_0,B_0)$ and $(C_0,A_0)$ are, respectively, controllable and observable, i.e.,}
\begin{align*}
    &\text{rank}\begin{bmatrix}B_0 & A_0 B_0 & \cdots & A_0^{n-1} B_0 \end{bmatrix}=n\\
    \text{and }\,~\,&\text{rank}\begin{bmatrix}C_0^{\top} & (C_0 A_0)^{\top} & \cdots & (C_0 A_0^{n-1})^{\top} \end{bmatrix}^{\top}=n\,.
\end{align*}
\end{assumption}

The power injection at each bus is composed of three major components: the generated power $P_g\!\in\!\mathbb{R}_{\geq 0}^m$, and the controllable load demand $P_{cd}\!\in\!\mathbb{R}_{\geq 0}^m$, and the non-controllable load demand $P_{ncd}\!\in\!\mathbb{R}_{\geq 0}^m$. Following the convention that a positive (respectively, negative) value of injection refers to net-generation (respectively, net-load), we have
\begin{align}
    P_{in} = P_g-\!P_{ncd}-\!P_{cd}\,.
\end{align} 
The controllable demand ($P_{cd}$) consists of various flexible loads, such as residential air-conditioners and electric water-heaters, which can be coordinated in hierarchical and distributed manner to provide grid stabilizing control support, as illustrated in \cite{lian2016hierarchical,moya2014hierarchical,kundu2017assessment,nandanoori2018prioritized,bhattacharya2019incentive,molina2011decentralized}. Generally, such control responsive loads can be modeled as:
\begin{align}\label{eq:controllable_power_input}
    P_{cd} = P_{cd,0} + K y\quad\text{(with $K\!\in\!\mathbb{R}^{m\times s}$)}\,,
\end{align}
where $P_{cd,0}\!\in\!\mathbb{R}_{\geq 0}^m$ are the nominal {(or, the baseline)} demand associated with the responsive loads, while $Ky$ represents the control-responsive part which responds to the grid measurements ($y$) as per the output feedback control gain matrix $K$. Unlike the generators, where the controller parameters can often be tuned accurately, the controller gain ($K$) in responsive loads is uncertain. For example, let us consider the case of distributed and hierarchical frequency control by a large number of residential thermostatic loads as in \cite{lian2016hierarchical,moya2014hierarchical,kundu2017assessment,nandanoori2018prioritized,bhattacharya2019incentive,molina2011decentralized}. Each thermostatic load receives a frequency threshold, and switches `on' (respectively, `off') if the measured frequency increases (respectively, decreases) beyond the associated threshold. Weather-related and end-use behavioral factors, however, drive uncertainties in the {load consumption and results in uncertainties in} the load response to those control thresholds \cite{kundu2017assessment,nandanoori2018prioritized}, which reflect in an uncertain aggregate controller gain matrix ($K$). 
% Modeling loads as stochastic is not new and studies \cite{qiu2008effects,milano2013systematic} discuss it in detail. 
% {\color{red}Umesh:  References justifying the use of stochastic uncertainty model for the loads will help at this point..} 
We model these demand-side control uncertainties via
\begin{subequations}\label{E:uncertain_K}
\begin{align}
    K &= \left(I - \Sigma\, \Xi\right)  K_0\,,\\
\text{with,}~\, \Xi &= \text{diag}(\xi_1, \dots,\xi_m),\,~\Sigma \!=\! \text{diag}(\sigma_1, \dots, \sigma_m)\,,
\end{align}
\end{subequations}
where $K_0\!\in\!\mathbb{R}^{m\times s}$ is the nominal control gain matrix obtained from design; $\Xi$ and $\Sigma$ are $m\!\times\! m$ diagonal matrices with $\lbrace \xi_i\rbrace_{i=1}^m$ and $\lbrace \sigma_i\rbrace_{i=1}^m$, respectively, as the diagonal entries; for each $i \in \{1,2,\dots,m\}$, $\xi_i = \frac{d\Delta_i}{dt}$ with $\{\Delta\}_{i=1}^m$ being the \textit{standard independent Wiener process} \cite{nandanoori2018mean} and each $\xi_i$ is a \textit{standard Gaussian white noise}, while $\sigma_i\!\geq\!0$ is a scaling factor for $\xi_i$ representing the associated standard deviation of the Gaussian white noise. Notice that the uncertainty modeled in \eqref{E:uncertain_K} appears as multiplicative to the system output ($y$) and, thereby, to the system states ($x$). Additionally, we allow uncertainties in the non-controllable load demand ($P_{ncd}$) as
\begin{align}\label{E:uncertain_ncd}
    P_{ncd} = P_{ncd,0} - \eta\,,\,\text{ with }\eta\!=\!\begin{bmatrix}\eta_1 & \dots &\eta_m\end{bmatrix}^\top\,,
\end{align}
where $P_{ncd,0}\!\in\!\mathbb{R}_{\geq 0}^m$ are the nominal values of the non-controllable load demand, and for each $i \in \{1,2,\dots,m\}$, $\eta_i = \frac{d\zeta_i}{dt}$ with $\{\zeta\}_{i=1}^m$ being the \textit{standard independent Wiener process} \cite{nandanoori2018mean} and each $\eta_i$ is a zero-mean \textit{standard Gaussian white noise}. Note here uncertainty models similar to \eqref{E:uncertain_ncd} can also be applicable to renewable generation (if any), as was shown in \cite{pushpak2018stochastic}. However, for the sake of simplicity in this paper, we consider the generated power ($P_g$) to be deterministic, while the uncertainties (additive and multiplicative) are modeled in the controllable and non-controllable loads, i.e.
\begin{subequations}
\begin{align}
P_{in}&=P_{in,0}-\!\left(I \!-\! \Sigma\, \Xi\right)  K_0\,C_0\,x+\eta\,,\\
\text{where }\,P_{in,0}&= P_g-\!P_{cd,0}-\!P_{ncd,0}\,.\label{E:power_nominal}
\end{align}
\end{subequations}

The resulting closed-loop system model {(shown in Fig.} \ref{fig:closed_loop_system_unc}) with stochastic uncertainties in load demand - both additive and multiplicative to the system states ($x$) - is given by
\begin{subequations}\label{eq_stochastic_sys_compact}
\begin{align}
\dot x &= A\, x + B \Sigma\, \Xi\, C\, x +  B P_{in,0} + B\, \eta\,,\\
\text{where }A\!&=\! A_0 \!-\! B_0\, K_0\, C_0\,, ~B\!=\!B_0\,, ~C \!=\! K_0\, C_0\,.
\end{align}\end{subequations}
\begin{figure*}[h]
\centering
\includegraphics[width = 0.75 \textwidth]{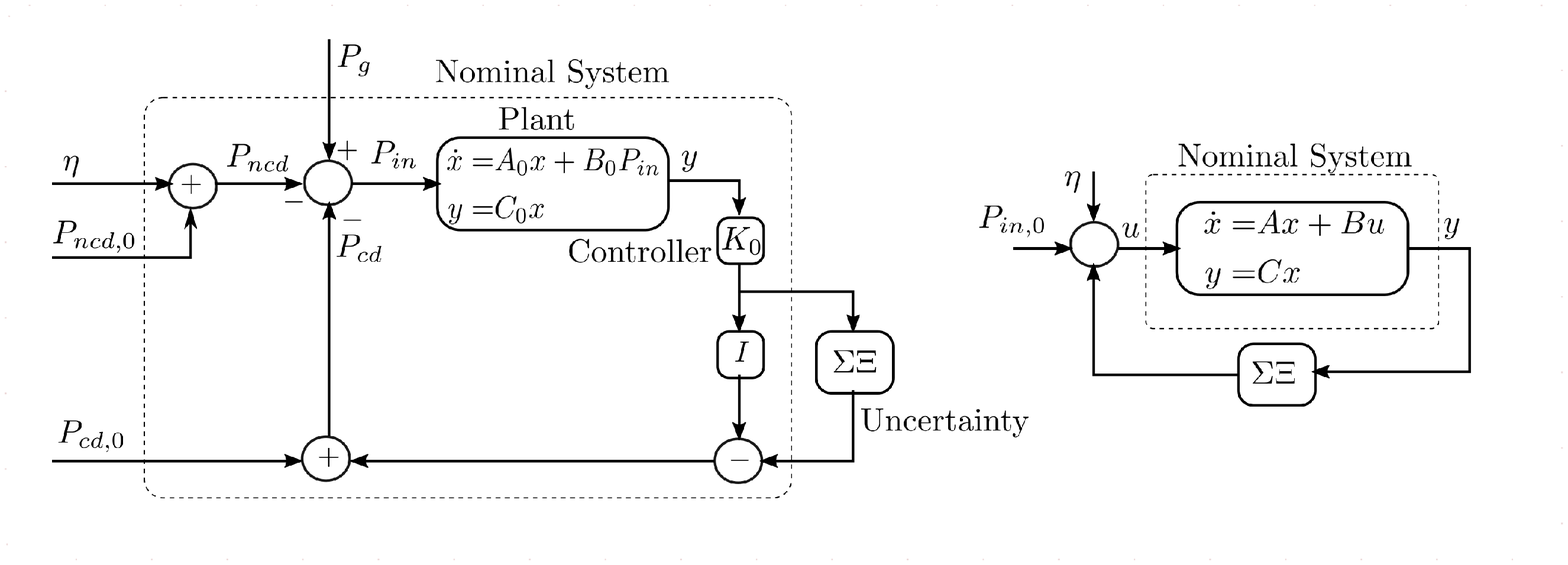}
\caption{The open loop system with uncertainties in the controllable and uncontrollable power is cast as a nominal system with multiplicative uncertainty in the feedback and additive uncertainty appearing as input with $A\!=\! A_0 \!-\! B_0\, K_0\, C_0\,, ~B\!=\!B_0\,, ~C \!=\! K_0\, C_0\, P_{in,0}= P_g-\!P_{cd,0}-\!P_{ncd,0}\,.$}
\label{fig:closed_loop_system_unc}
\end{figure*}

The term $B P_{in,0}$ acts as a constant forcing to the system. Under Assumption \ref{AS:sys}, when $A$ is Hurwitz, we can choose a new coordinate $z\!=\!x\!+\!A^{-1}BP_{in,0}$, and rewrite \eqref{eq_stochastic_sys_compact} as:
\begin{align}\label{eq_stochastic_sys_compact_z}
\dot z &= A\, z + B \Sigma\, \Xi\, C\, z - B \Sigma\, \Xi\,C A^{-1} B P_{in,0} + B\, \eta\,.
\end{align}
Notice that the uncertainty appears as both multiplicative to the system states (the second term on the right hand side), as well as additive (the third and the fourth terms).

\section{Problem Description}
\label{Sec_problem}

{This section presents the necessary definitions of stochastic stability and formally describes the problem statement on design of a sparse controller for the linear stochastic power network in the presence of stochastic uncertainties. }

\vspace{-0.1in}\subsection{Stochastic Stability Analysis}
\label{Subsec_stochastic_stability}

The conventional notions of transient stability do not apply to the stochastic dynamic models of the power systems. Instead, two alternative notions of stability, namely the \textit{second moment bounded stability} (SMBS) and the \textit{mean square exponential stability} (MSES), are more relevant \cite{nandanoori2018mean,pushpak2018stochastic}. In contrast to transient stability which is concerned with the convergence of the power systems states (e.g. frequency), these alternative notions of stochastic stability are concerned with the variances of the power system states. As shown in \cite{nandanoori2018mean,pushpak2018stochastic}, such notions of stochastic stability allow quantification of the allowable uncertainties and hence adopted in this work.

\begin{definition} \cite{nandanoori2018mean} The stochastic system in \eqref{eq_stochastic_sys_compact_z} is second moment bounded stable (SMBS) if there exists a positive scalar $\gamma_0$ such that $\lim_{t \to \infty} \mathbb{E}\left[\| z(t) \|^2 \vert z(0) \right]\!\leq\! \gamma_0,$
for all initial conditions $z(0)$; where the expectation (denoted by $\mathbb{E}$) is taken over the stochastic processes $\lbrace \xi_i(t)\rbrace_{i=1}^m$ and $\lbrace\eta_i(t)\rbrace_{i=1}^m$. 
\end{definition}
Now let us consider another stochastic dynamical system:
\begin{align}\label{eq_stochastic_sys_compact_ztilde}
\dot {\widetilde{z}} &= A\, \widetilde{z} + B \Sigma\, \Xi\, C\, \widetilde{z} \,,\,\quad \widetilde{z}(0)\!=\!z(0)\,,
\end{align}
which is modified from \eqref{eq_stochastic_sys_compact_z} by removing the terms associated with additive uncertainties, while keeping the same initial condition. The following notion of MSES applies to \eqref{eq_stochastic_sys_compact_ztilde}.
\begin{definition} \cite{nandanoori2018mean}
The stochastic system in \eqref{eq_stochastic_sys_compact_ztilde} is mean square exponentially stable (MSES) if there exist positive scalars $\gamma_1,\gamma_2,$ such that $\mathbb{E}\left[ \| \widetilde{z}(t) \|^2  \vert \widetilde{z}(0) \right]\!\leq\! \gamma_1 e^{-\gamma_2 t}  \| \widetilde{z}(0) \|^2 ~\forall t$\,, for every $\widetilde{z}(0)$, where the expectation is taken over $\lbrace \xi_i(t)\rbrace_{i=1}^m$.
\end{definition}
Recently developed results in \cite{nandanoori2018mean,pushpak2018distributed} establish the equivalence between the above notions of stability. In particular, the authors argue \cite[Remark 4]{nandanoori2018mean} that the SMBS of the stochastic dynamical system in \eqref{eq_stochastic_sys_compact_z} is equivalent to the MSES of the stochastic dynamical system in \eqref{eq_stochastic_sys_compact_ztilde}. With this equivalence in mind, we will henceforth only focus on the notion of MSES for the stochastic dynamical system in \eqref{eq_stochastic_sys_compact_ztilde}, noting that the derived results also apply in the context of SMBS for the original system \eqref{eq_stochastic_sys_compact_z}. In addition to Assumption\,\ref{AS:sys}, we require the following to hold for the MSES results:
%
%---------
\begin{assumption}
\label{assumptions} \cite{nandanoori2018mean}
The initial condition $\widetilde{z}(0)$ has a bounded variance, and is independent of the processes $\lbrace \xi_i(t)\rbrace_{i=1}^m$\,. 
\end{assumption}

%%%%%%%%%%%%%%%%%%%%%%%%%%%%%%%%%%%%%%%%%%%%%%%%%%%%%%%%
%%%%%%%%%%%%%%%%%%%%%%%%%%%%%%%%%%%%%%%%%%%%%%%%%%%%%%%%
%%%%%%%%%%%%%%%%%%%%%%%%%%%%%%%%%%%%%%%%%%%%%%%%%%%%%%%%
%%%%%%%%%%%%%%%%%%%%%%%%%%%%%%%%%%%%%%%%%%%%%%%%%%%%%%%%
%%%%%%%%%%%%%%%%%%%%%%%%%%%%%%%%%%%%%%%%%%%%%%%%%%%%%%%%
%%%%%%%%%%%%%%%%%%%%%%%%%%%%%%%%%%%%%%%%%%%%%%%%%%%%%%%%
%%%%%%%%%%%%%%%%%%%%%%%%%%%%%%%%%%%%%%%%%%%%%%%%%%%%%%%%
\vspace{-0.1in}\subsection{Problem Statement}
Let us denote, for every $i\!\in\!\lbrace 1,\dots,m\rbrace$, the $i^\text{th}$ column of $B$ by $B_i\!\in\!\mathbb{R}^{n\!\times\! 1}$, and the $i^\text{th}$ row of $C$ by $C_i\!\in\!\mathbb{R}^{1\!\times\! n}$. Therefore, we can rewrite \eqref{eq_stochastic_sys_compact_ztilde} as:
\begin{subequations}\label{eq_stochastic_sys_mult}
\begin{align}
    \dot {\widetilde{z}} &= A\, \widetilde{z} +  \sum_{i=1}^m\sigma_i\, B_i\,C_i\, \widetilde{z}\,\xi_i\,,\\
    \text{where }A\!&=\! A_0 \!-\! B_0\, K_0\, C_0\,, ~B\!=\!B_0\,, ~C \!=\! K_0\, C_0\,.
\end{align}
\end{subequations}

In this paper, we are concerned with designing an output feedback controller gain matrix $K_0$ such that the system \eqref{eq_stochastic_sys_mult} is MSES for some (maximally) allowable uncertainties in the controllable loads, given by the standard deviations $\lbrace \sigma_i\rbrace_{i=1}^m$ associated with each of the stochastic processes $\lbrace \xi_i(t)\rbrace_{i=1}^m$\, {(which are modeled in the controllable loads)}. 

\begin{remark} \label{rem:K_caveat}
{A trivial solution of} $K_0\!=\!0$ {ensures MSES of the system} \eqref{eq_stochastic_sys_mult}, {while admitting extremely large (theoretically unbounded) multiplicative uncertainty levels since the deterministic system}  {is stable, that is,} $A_0$ {is Hurwitz (as per Assumption}\,\ref{AS:sys}). 
\end{remark}
{Following Remark} \ref{rem:K_caveat}, {the challenging aspect of this work is to propose a well-posed control synthesis problem that can be solved for a feedback controller} $K_0$. {Furthermore, the control synthesis problem must address practical power system considerations such as control under partial observation, improved transient performance and reduce the cost and deployment of the developed controls by promoting sparse control design. }

\section{Stochastic Stabilizing Control Synthesis}
\label{sec_robust_control}

We start by applying results from \cite{nandanoori2018mean} into deriving the following MSES stability conditions for the system in \eqref{eq_stochastic_sys_mult}.
\begin{proposition}
The closed-loop stochastic system \eqref{eq_stochastic_sys_mult} is MSES if and only if there exist a positive definite $\mathcal{Q}\!\in\!\mathbb{R}^{n\times n}$, an $M\!\in\!\mathbb{R}^{m\times n}$ and positive scalars $\lbrace\alpha_i\rbrace_{i=1}^m,\lbrace\sigma_i\rbrace_{i=1}^m$ so that
\begin{subequations}\label{E:proposition}
\begin{align}
    &A_0 \mathcal{Q} \!+\! \mathcal{Q} A_0^{\top} \!-\! B_0 M \!-\! M^\top B_0^\top \!+\! \displaystyle \sum_{i=1}^m \alpha_i B_{0_i}  B_{0_i}^{\top} \!\prec 0
    \\
    & \sigma_i^2 M_i \mathcal{Q}^{-1} M_i^{\top} < \alpha_i \quad\forall i \!\in\!\lbrace 1, \dots, m\rbrace. 
\end{align}\end{subequations}
where $M_i\!\in\!\mathbb{R}^{1\times n}$ is the $i^{\text{th}}$ row vector of the matrix $M$, and $B_{0_i}\!\in\!\mathbb{R}^{n\times 1}$ is the $i^{\text{th}}$ column vector of $B_0$\,.
\label{cor:MIMO_system_stability}
\end{proposition}
\begin{proof}
As per \cite[Theorem 11]{nandanoori2018mean}, the system \eqref{eq_stochastic_sys_mult} is MSES if and only if there exists a positive definite $\mathcal{P}\!\in\!\mathbb{R}^{n\times n}$ satisfying
\begin{align}
    A^{\top} \mathcal{P} + \mathcal{P} A + \sum_{i=1}^m \sigma_i^2 C_i^{\top} B_i^{\top} \mathcal{P} B_i C_i \prec 0
    \label{eq:stability_Lyapunov_eq}
\end{align}
The dual of the condition \eqref{eq:stability_Lyapunov_eq} is given by
{\begin{align}
    A \mathcal{Q} + \mathcal{Q} A^{\top} + \sum_{i=1}^m \sigma_i^2 B_i C_i \mathcal{Q}\, C_i^{\top} B_i^{\top} \prec 0
    \label{eq:stability_Lyapunov_dual_eq}
\end{align}
for some positive definite matrix $\mathcal{Q}\!\in\!\mathbb{R}^{n\times n}$.} On application of \cite[Lemma 12]{nandanoori2018mean}, the matrix inequality in \eqref{eq:stability_Lyapunov_dual_eq} can be split into the following two equivalent inequalities 
\begin{subequations}\label{E:split_ineq}
\begin{align}
    & A \mathcal{Q} + \mathcal{Q} A^{\top} + \sum_{i=1}^m \alpha_i B_i B_i^{\top} \prec 0 \\
    & \sigma_i^2 \,C_i\, \mathcal{Q}\, C_i^{\top} < \alpha_i, \quad\forall i \!\in\!\lbrace 1, \dots, m\rbrace\,,\label{eq:bilinear}
\end{align}
\end{subequations}
by introducing positive scalars $\lbrace \alpha_i\rbrace_{i=1}^m$ as new variables. The proof is complete after we substitute for $A,B,C$ from \eqref{eq_stochastic_sys_mult} in the above inequalities \eqref{E:split_ineq}, and eliminate $K_0$ with a new matrix variable {$M \!:=\! K_0\,C_0\,\mathcal{Q}$}.
\end{proof}

{In view of Remark} \ref{rem:K_caveat}, {the controller synthesis problem becomes well posed, once we enforce certain closed-loop transient performance goals for the nominal deterministic system (with} $\sigma_i\!=\!0$ in \eqref{eq_stochastic_sys_mult}). {In particular, consider a Lyapunov function candidate} $\Psi\!=\!\widetilde{z}^\top \mathcal{P}\widetilde{z}$\,, {for some} $\mathcal{P}\!\succ\!0$ \cite{slotine1991applied}. {Then the closed-loop transient performance goal of the nominal deterministic system} $\dot {\widetilde{z}}\!=\! A\, \widetilde{z} $ {can be reformulated as identifying an acceptably large positive scalar} $\gamma$ {such that}
\begin{align}
    A^\top\mathcal{P}+\mathcal{P}A\prec-\gamma\,I_n\,,
\end{align}
{where} $I_n$ {denotes an identity matrix of size} $n$\,.

The inequalities in \eqref{E:proposition} form a nonconvex feasibility problem with optimization variables $\mathcal{Q},M,\lbrace \alpha_i\rbrace_{i=1}^m$ and $\lbrace \sigma_i^2 \rbrace_{i=1}^m$\,, due to the bilinear constraint \eqref{eq:bilinear}. In the following, we present convex reformulations of \eqref{E:proposition} for synthesizing stochastic stabilizing optimal controller, while maximizing the allowable uncertainties at the controllable loads {and incorporating the transient performance constraints}.

%----------------------------------------
\vspace{-0.1in}\subsection{Full-State Measurement (or, $\text{rank}(C_0)\!=\!n$)}
%----------------------------------------
At first, let us consider the special case of state feedback control synthesis, assuming full-state measurement. Full-state measurement refers to the scenario when the original state vector ($x$) can be uniquely determined from the measurements ($y\!=\!C_0\,x$), i.e. $\text{rank}(C_0)\!=\!n$\,. A special case of this is when $C_0\!=\!I_n$\,.
Let us consider the convex optimization problem:
\begin{subequations}\label{eq:controller_synthesis}
\begin{align}
     \underset{\gamma_1, \gamma_2, \lbrace\beta_i\rbrace_{i=1}^m,\widetilde{\mathcal{Q}},\widetilde{M}}{\text{minimize}}\hspace{0.4in}
    -w_1\gamma_1 - w_2\gamma_2 + w_3\!\sum_{i=1}^m \beta_i \\
    \text{subject to}, \hspace{1.2in} I_n\succeq \widetilde{\mathcal{Q}}\succ \gamma_1I_n\,,\label{eq:Q_constr}\\
    \,~\, \begin{bmatrix}
    \beta_i & \widetilde{M}_i \\
    \widetilde{M}_i^{\top} & \widetilde{\mathcal{Q}}
    \end{bmatrix} \succ 0\,~\, \forall i\!\in\!\lbrace 1,\dots,m\rbrace,  \label{eq:lmi_feas1}\\
    A_0 \widetilde{\mathcal{Q}} +\! \widetilde{\mathcal{Q}} A_0^{\top} \!\!-\! B_0 \widetilde{M} \!-\! \widetilde{M}^{\top}\!B_0^\top \!+\! \displaystyle\sum_{i=1}^m\! B_{0_i} B_{0_i}^{\top} \!\prec -\gamma_2I_n\,, \label{eq:performance}
\end{align}\end{subequations}
where $\lbrace\beta_i\rbrace_{i=1}^m$ and $\gamma_{1},\gamma_2$ are positive scalars, $\widetilde{\mathcal{Q}}\!\in\!\mathbb{R}^{n\times n}$ is a positive definite matrix, $\widetilde{M}\!\in\!\mathbb{R}^{m\times n}$, and $w_1\!>\!0,w_2\!>\!0,w_3\!>\!0$ are scalar weights. {The choice of the scalar weights determine the trade-off between the sparsity and performance costs.} We denote the optimal solution of \eqref{eq:controller_synthesis} by the tuple
\begin{align*}
    \left(\gamma_1^*,\gamma_2^*,\lbrace \beta_i^*\rbrace_{i=1}^m,\mathcal{Q}^*,M^*\right)\,.
\end{align*}
\begin{theorem} (\textsc{Full-Rank $C_0$})
Consider the optimal solution $\!\left(\gamma_1^*,\gamma_2^*,\lbrace \beta_i^*\rbrace_{i=1}^m,\mathcal{Q}^*\!,M^*\right)$ of the problem \eqref{eq:controller_synthesis}. Under full-state observation, i.e. $\text{rank}(C_0)\!=\!n$\,, a controller with gain 
\begin{align}\label{eq:K_uniq}
        K_0\!=\!M^*\left(\mathcal{Q}^*\right)^{-1}\left(C_0^\top C_0\right)^{-1}C_0^\top
    \end{align}
guarantees MSES of \eqref{eq_stochastic_sys_mult} with allowable uncertainties of $\sigma_i\!<\!1/\!\sqrt{\beta_i^*}$ for every $i\!\in\!\lbrace 1,\dots,m\rbrace$, with a closed-loop transient performance certificate as $A\mathcal{Q}^*\!+\mathcal{Q}^*\!A^{\!\top}\!\!\prec\!-\gamma_2^*I_n$\,.
\label{thm:main_thm}
\end{theorem}
\begin{proof}
Note that, for a positive scalar $\widetilde{\alpha}$\,, the inequalities:
\begin{subequations}\label{E:proposition_recast}
\begin{align}
    \!\!\!&A_0 \mathcal{Q} \!+\! \mathcal{Q} A_0^{\top} \!-\! B_0 M \!-\! M^\top\! B_0^\top \!+ \displaystyle \widetilde{\alpha}\sum_{i=1}^m\! B_{0_i}  B_{0_i}^{\top} \!\prec\! -\gamma I_n\!\!\\
    \!\!\!& \sigma_i^2 M_i \mathcal{Q}^{-1} M_i^{\top} < \widetilde{\alpha} \quad\forall i \!\in\!\lbrace 1, \dots, m\rbrace.\!\!
\end{align}\end{subequations}
form a special case of the conditions \eqref{E:proposition} in Proposition\,\ref{cor:MIMO_system_stability}, with $\alpha_i\!=\!\widetilde{\alpha}$ for every $i\!\in\!\lbrace 1,\dots,m\rbrace$\,. Therefore, the inequalities in \eqref{E:proposition_recast} are sufficient conditions for MSES of the system \eqref{eq_stochastic_sys_mult}. Defining $\widetilde{\mathcal{Q}}\!:=\!\mathcal{Q}/\widetilde{\alpha},\,\widetilde{M}\!:=\!M/\widetilde{\alpha},\,\beta_i\!:=\!1/\sigma_i^2~\forall i\!\in\!\lbrace 1,\dots,m\rbrace$, and $\gamma_2\!:=\!\gamma/\widetilde{\alpha}$, we can recast \eqref{E:proposition_recast} into:
\begin{subequations}\label{E:proposition_tilde}
\begin{align}
    &A_0 \widetilde{\mathcal{Q}} +\! \widetilde{\mathcal{Q}} A_0^{\top} \!\!-\! B_0 \widetilde{M} \!-\! \widetilde{M}^{\top}B_0^\top \!+\! \sum_{i=1}^m B_{0_i} B_{0_i}^{\top} \!\prec\! -\gamma_2\,I_n\label{eq_performance}\\
    & \widetilde{M}_i \widetilde{\mathcal{Q}}^{-1} \widetilde{M}_i^{\top} < \beta_i \quad\forall i \!\in\!\lbrace 1, \dots, m\rbrace. \label{eq:lmi_feas2}
\end{align}\end{subequations}
Since $\widetilde{\mathcal{Q}}\!\succ\! 0$, we can apply Schur complement to establish the equivalence between \eqref{eq:lmi_feas2} and the matrix inequality in \eqref{eq:lmi_feas1}. Therefore, any $K_0$ satisfying $K_0C_0\!=\!M^*\!\left(\mathcal{Q}^*\right)\!^{-1}\!$ guarantees MSES of \eqref{eq_stochastic_sys_mult}. Moreover, since $\text{rank}\left(C_0\right)\!=\!n$, $C_0^\top C_0$ is invertible, and a unique $K_0$ can be obtained using the \textit{pseudo-inverse} of $C_0$ as in \eqref{eq:K_uniq}\,. Additionally, it is easy to check that every $\sigma_i\!=\!1/\sqrt{\beta_i}\leq\!1/\sqrt{\beta_i^*}$ satisfies the condition \eqref{eq:lmi_feas2} with $M^*$ and $\mathcal{Q}^*$\,. Henceforth we will refer to 
\begin{align} \label{eq:critical_unc_full_obs}
    \sigma_i^*\!:=\!1/\sqrt{\beta_i^*}~\forall i \!\in\!\lbrace 1, \dots, m\rbrace
\end{align} 
as the \textit{critical uncertainties} (maximally) allowable at the controllable loads for guaranteed closed-loop stochastic stability.

Note that the problem is ill-defined (unbounded) without the inequality constraint $I_n\!\succeq\!\widetilde{Q}$ in \eqref{eq:Q_constr}. Since $A_0$ is Hurwitz (Assumption\,\ref{AS:sys}), there exists a $\mathcal{Q}\!\succ\!0$ such that $A_0 \mathcal{Q} +\! \mathcal{Q} A_0^{\top}\!\!\prec\!0$\,. It is easy to see that the choice of $\widetilde{M}\!=\!0_{m\times n}$ and $\widetilde{\mathcal{Q}}\!=\!l\mathcal{Q}$, for any arbitrarily large scalar $l\!>\!0$\,, will satisfy both the conditions in \eqref{E:proposition_tilde}. The constraints \eqref{eq:Q_constr} ensure well-posedness of the problem by ensuring boundedness of the solutions as well as ensuring well-conditioned $\mathcal{Q}^*$\,. In particular, the minimum eigenvalue of the optimal $Q^*$ satisfies $\lambda_{\min}(\mathcal{Q}^*)\!>\!\gamma_1^*$. Since $\widetilde{M}_i \widetilde{\mathcal{Q}}^{-1} \widetilde{M}_i^{\top}\!\leq\!\|\widetilde{M}_i\|_2^2/\lambda_{\min}(\widetilde{\mathcal{Q}})$, we can establish the following bound on the critical uncertainties:
\begin{align}\label{eq:pareto}
    \forall i:\,~\beta_i^*\!\leq\!\|M_i^*\|_2^2/\lambda_{\min}(\mathcal{Q}^*)\!\implies\! \sigma_i^*\!>\!\sqrt{\gamma_1^*}/\|M_i^*\|_2\,.
\end{align}

Finally, since each $B_{0_i} B_{0_i}^{\top}\!\!\succeq\! 0$, we immediately note from \eqref{eq_performance} that $A\mathcal{Q}^*\!+\!\mathcal{Q}^*\!A^{\!\top}\!\!\prec\!-\gamma_2^*I_n$. With $\!\mathcal{P}\!\!=\!\!\left(\mathcal{Q^*}\right)\!^{\!-\!1}\!\!$, we get the dual condition: $A^{\!\top}\!\mathcal{P}\!+\!\mathcal{P}\!A\!\prec\!-\gamma_2^*I_n$. Using the Lyapunov stability theory \cite{slotine1991applied}, this certifies for the nominal deterministic system $\dot{\widetilde{z}}\!=\!A\widetilde{z}$ a guaranteed rate of exponential convergence to the nominal equilibrium given by $\dot{\Psi}\!<\!-\gamma_2^*\|\widetilde{z}\|_2^2$, for some Lyapunov function $\Psi\!=\!\widetilde{z}^{\top}\mathcal{P}\widetilde{z}$. This completes the proof.
\end{proof}

\textbf{Fundamental Control Trade-Off:} The relation \eqref{eq:pareto} illustrates a fundamental trade-off between the optimal control efforts, and the allowable (critical) uncertainty levels at each controllable load bus. This is an artifact of the multiplicative uncertainties considered in this paper, in which the impact of the uncertainties on the system dynamics are accentuated via high gain controls. Thus we have a Pareto curve of optimal solutions trading off control efforts with allowable uncertainties, as illustrated in Sec.\,\ref{Sec_simulation}. This is unlike the scenario with additive uncertainties \cite{moya2014hierarchical}, in which a high gain controller is typically expected to admit larger uncertainties.

\textbf{Sparsity-Promoting Control Design:} The control synthesis problem \eqref{eq:controller_synthesis} can be easily reformulated to design a row-sparse controller gain ($K_0$). Row-sparsity in $K_0$ is a desirable property since it allows the grid operator to engage as few load control inputs as possible. Since a row-sparse optimal $M^*$ also implies a row-sparse optimal $K_0$\,, it is sufficient to promote row-sparsity of $M^*$ in a reformulation of \eqref{eq:controller_synthesis}. Following an approach similar to \cite{polyak2013lmi}, this can be achieved by adding a weighted $l_{1,p}$-norm of $\widetilde{M}$ in the objective in \eqref{eq:controller_synthesis}, defined as
\begin{align}\label{eq:M_norm}
    \!\!\text{($l_{1,p}$-norm)}\,~\,\|\widetilde{M}\|_{1,p}:=\!\!\sum_{i=1}^m\|\widetilde{M}_i\|_p\, && \text{for } p\!\in\!\lbrace 1,2,\infty\rbrace\,.\!\!
\end{align}
Therefore, a sparsity-promoting reformulation of the optimal control design problem \eqref{eq:controller_synthesis} is given by:
\begin{subequations}\label{eq:controller_synthesis_sparse}
\begin{align}
     \!\!\underset{\gamma_1, \gamma_2, \lbrace\beta_i\rbrace_{i=1}^m,\widetilde{\mathcal{Q}},\widetilde{M}}{\text{minimize}}\hspace{0.2in}\begin{bmatrix}w_1~w_2~w_3\end{bmatrix}\begin{bmatrix}-\gamma_1\\-\gamma_2\\\displaystyle\sum_{i=1}^m \beta_i\end{bmatrix}+\|\widetilde{M}\|_{1,p} \\
    \!\!\text{subject to}, \hspace{0.95in}\eqref{eq:Q_constr}\,,~ \eqref{eq:lmi_feas1}\,,~ \eqref{eq:performance}\,,
\end{align}\end{subequations}
for different choices of $p\in\!\lbrace 1,2,\infty\rbrace$\,.

%----------------------------------------
\vspace{-0.1in}\subsection{Partial-State Measurement (or, $\text{rank}(C_0)\!<\!n$)}\label{sec:control_partial}
%----------------------------------------

In most power systems today, only limited number of buses are equipped with synchrophasors giving access to only partial-state measurements (e.g. phase angles and frequencies at generator buses alone). In such a case, $\text{rank}(C_0)\!<\!n$ and Theorem\,\ref{thm:main_thm} does not apply. While it is possible to identify the best $K_0$ by minimizing $\|M^*\!-\!K_0\,C_0\,Q^*\|_F$ (where $\|\cdot\|_F$ is the Frobenius norm), the equality $M^*\!=\!K_0\,C_0\,Q^*$ may not hold exactly, thereby not guaranteeing the MSES conditions. Instead, we adopt the following two-stage control synthesis approach. It is worth noting that such approaches often offer attractive solutions to complex control applications such as polynomial control design for power networks \cite{kundu2015stability,kundu2019distributed}.

\textbf{Stage 1 (Pre-computation):} In the first stage, we solve a problem similar to \eqref{eq:controller_synthesis_sparse}, but with an additional constraint to enforce negative-definiteness of $A_0 \widetilde{\mathcal{Q}} +\! \widetilde{\mathcal{Q}} A_0^{\top}$\,, as follows:
\begin{subequations}\label{eq:controller_synthesis_stage1}
\begin{align}
     \!\!\underset{\gamma_1, \gamma_2, \lbrace\beta_i\rbrace_{i=1}^m,\widetilde{\mathcal{Q}},\widetilde{M}}{\text{minimize}}\hspace{0.2in}\begin{bmatrix}w_1~w_2~w_3\end{bmatrix}\begin{bmatrix}-\gamma_1\\-\gamma_2\\\displaystyle\sum_{i=1}^m \beta_i\end{bmatrix}+\|\widetilde{M}\|_{1,p} \\
    \!\!\text{subject to}, \hspace{1.15in}A_0 \widetilde{\mathcal{Q}} +\! \widetilde{\mathcal{Q}} A_0^{\top}\!\prec 0\,,\label{eq:open_loop_stable}\\
    \eqref{eq:Q_constr}\,,~ \eqref{eq:lmi_feas1}\,,~ \eqref{eq:performance}\,,
\end{align}\end{subequations}
The additional constraint \eqref{eq:open_loop_stable} allows us to guarantee the existence of a feasible control solution in the second stage, as would be made clear next. The first stage problem \eqref{eq:controller_synthesis_stage1} only serves as pre-computation stage to find an optimal $\mathcal{Q}^*$ which would be used in the second stage.

\textbf{Stage 2 ($K_0$ Synthesis):} The optimal $\mathcal{Q}^*$ found in the pre-computation stage is used in the second-stage optimization problem to identify the optimal stabilizing control gain $K_0$\,, along with an optimal $\gamma_2^*$, as follows:
% 
% \begin{subequations}
\begin{align}\label{eq:controller_synthesis_stage2}
     \!\!\!\!\underset{\gamma_2,\widetilde{K}}{\text{minimize}}\hspace{0.05in}~-\!w_2\gamma_2\!+\!w_3\|\widetilde{K}C_0\sqrt{\mathcal{Q}^*}\|_{1,2} \!+\|\widetilde{K}C_0\mathcal{Q}^*\|_{1,p}\!\\
    \!\!\!\!\text{subject to}, \hspace{0.2in}\begin{aligned} \left(A_0\!-\!B_0 \widetilde{K}C_0\right)\! \mathcal{Q}^*\!\! +\! \mathcal{Q}^*\! \!\left(A_0\!-\!B_0 \widetilde{K}C_0\right)\!^\top \\
    +\sum_{i=1}^m\! B_{0_i} B_{0_i}^{\top} \!\prec -\gamma_2I_n~
    \end{aligned},\!\!\!\!\notag
\end{align}
% \end{subequations}
%
Let us denote the optimal solution of \eqref{eq:controller_synthesis_stage2} by $\!\left(\gamma_2^*,K_0\right)$\,.

\begin{theorem}(\textsc{Partial-Rank $C_0$})
The control synthesis problem \eqref{eq:controller_synthesis_stage2} is feasible. Moreover, the optimal controller $K_0$ guarantees MSES of \eqref{eq_stochastic_sys_mult} with allowable uncertainties of 
\begin{align}\label{eq:sigma}
\sigma_i\!<\!\left(\|K_{0_i}C_0\sqrt{\mathcal{Q}^*}\|_2\right)^{-1}   \quad \forall i\!\in\!\lbrace 1,\dots,m\rbrace
\end{align}
where $K_{0_i}$ denotes the $i^\text{th}$ row of $K_0$\,, achieving a closed-loop transient performance certificate as $A\mathcal{Q}^*\!+\mathcal{Q}^*\!A^{\!\top}\!\!\prec\!-\gamma_2^*I_n$\,.
\end{theorem}

\begin{proof}
Since pre-computed $\mathcal{Q}^*$ satisfies the Lyapunov-like stability constraint \eqref{eq:open_loop_stable}, we can use Assumption\,\ref{AS:sys} to guarantee the existence of a control gain $\widetilde{K}$ that satisfies the constraint in \eqref{eq:controller_synthesis_stage2} for some $\gamma_2\!\geq\!0$\,. The rest of the proof follows similar arguments as in Theorem\,\ref{thm:main_thm} after we replace $\widetilde{M}$ by $\widetilde{K}C_0\mathcal{Q}^*$\!. Note the addition of the sparsity-promoting $l_{1,p}$-norm\footnote{The sparsity can be enforced directly on $K$ in \eqref{eq:controller_synthesis_stage2}. But we enforced it on $KC_0\mathcal{Q}^*$ to maintain uniformity with the earlier formulation in \eqref{eq:controller_synthesis_sparse}.} on $\widetilde{M}\!=\!\widetilde{K}C_0\mathcal{Q}^*$\!. Moreover, by the definition of the square root of a matrix, we have $\mathcal{Q}^*\!=\!\sqrt{\mathcal{Q}^*}\sqrt{\mathcal{Q}^*}^\top$\!. Therefore, the constraint \eqref{eq:lmi_feas2} from Theorem\,\ref{thm:main_thm} can be reformulated as
\begin{align*}
    \sigma_i^2\widetilde{K}_iC_0\sqrt{\mathcal{Q}^*}\sqrt{\mathcal{Q}^*}^\top C_0^\top \widetilde{K}_i^\top\!\!<\! 1\!\implies\!\sigma_i\|\widetilde{K}_iC_0\sqrt{\mathcal{Q}^*}\|_2\!<\!1
\end{align*}
where $\widetilde{K}_i$ is the $i^\text{th}$ row of $\widetilde{K}$. Thus we eliminate the variables $\beta_i$ and the associated constraints \eqref{eq:lmi_feas1}; and replace the penalty on $\beta_i$ by a penalty on $\|K_iC_0\sqrt{\mathcal{Q}^*}\|_2$ in order to maximize the allowable uncertainties. This completes the proof.
\end{proof}

{It is important to reiterate that the proposed sparse control synthesis framework under partial observation is very generic and does not make any assumptions on the location or minimal number of PMU devices. On the contrary, the  optimal PMU placement is a well-studied problem and recent works such as} \cite{gou2008optimal,hajian2011optimal,vaidya2013optimal,peppanen2012optimal,akhlaghi2016optimal} {guarantee the observability of the system with minimum number of PMUs. Combining these works on optimal placement of PMUs and the sparse control synthesis framework will be an interesting direction to pursue.}

\section{Numerical Example}
\label{Sec_simulation}

%--------------------
% ALL the figures
%--------------------
MATLAB-based Power System Toolbox \cite{PST} is used to illustrate the proposed stochastic control framework on the IEEE 39-bus test system \cite{moeini2015open} consisting of 10 generators and 19 loads. Here we briefly describe the network preserving model with linear frequency dependent real power loads \cite{bergen1981structure}, which is used for the numerical studies. Swing equations are used to model the synchronous generator dynamics, while loads (primarily driven by induction motors) are modeled as affinely dependent on the frequencies. In the augmented network model including the internal buses for the generators \cite{bergen1981structure}, the (internal) generator buses are denoted by $\mathcal{G}$ and the load buses as $\mathcal{L}$\,. Moreover, due to the internal bus representation, the generators and loads in the augmented network are not connected to the same bus, i.e. $\mathcal{G}\!\cap\!\mathcal{L}\!=\!\emptyset$\,. Compactly, the power system dynamic model is given by: 
\begin{subequations}\label{E:power_dyn}
\begin{align}
  \!\!\!\!M_i \dot{\omega}_{i}  \!+\!D _i \dot{\delta}_i  &=  P_{g_i}   - P_{e_i}(\delta) &\forall i\!\in\!\mathcal{G},\!\!\\
  D _i   \dot{\delta}_{i}  &=   -  P_{d_i}  - P_{e_i}(\delta) &\forall i\!\in\!\mathcal{L},\!\!\\
  P_{e_i}(\delta)&=\!\sum_j\!E_i E_jY_{ij}\cos\left(\delta_i\!-\!\delta_j\!-\!\theta_{ij}\right)&\forall i,\!\!
\end{align}\end{subequations}
where $E_i$ and $\delta_i$ are, respectively, the magnitude and phase angle of the $i^\text{th}$ bus voltages; $\delta$ is the vector of all bus angles ($\delta_i$); $Y_{ij}$ and $\theta_{ij}$ are, respectively, the modulus and the phase angle of the transfer admittance between buses $i$ and $j$; while $P_{e_i}(\delta)$ is the electrical power injected into the network at bus-$i$\,. For each generator at bus $i\!\in\!\mathcal{G}$\,, $M_i\!>\!0$ is the inertia constant, $ D_i\!>\! 0   $ is the damping coefficient, and $ P_{g_i}\!>\!0 $ is the mechanical power input. For each load at bus $i\!\in\!\mathcal{L}$\,, $D_i \!>\! 0$ is the load-frequency coefficient and $P_{d_i} \!>\! 0  $ is the load demand at equilibrium. As described in Sec.\,\ref{Sec_model}, the loads consist of controllable and non-controllable parts, i.e. $P_{d_i}\!=\!P_{cd_i}\!+\!P_{ncd_i}$\,.

After the adjustment for the reference bus (indexed by `R') phase angle \cite{bergen1981structure}, we can express \eqref{E:power_dyn} more compactly as: 
\begin{align}\label{E:power_compact}
    &\dot{x}= F(x)+B\,P_{in}\,,\\ \text{with}~&x\!:=\!\left(\lbrace \omega_i\rbrace_{i\in\mathcal{G}}\,,\, \lbrace \delta_i\!-\!\delta_R\rbrace_{i\in\mathcal{G}\cup\mathcal{L}\backslash\lbrace R\rbrace}
    \right)^{\!\top}\!\!, ~P_{in}\!=\!P_g\!-\!P_d\,,\notag
\end{align}
where $x\!\in\!\mathbb{R}^n$, $P_{in}\!\in\!\mathbb{R}^m$, $F\!:\!\mathbb{R}^n\!\mapsto\!\mathbb{R}^n$ is locally Lipschitz, and $B\!\in\!\mathbb{R}^{n\times m}$. We obtain \eqref{E:sys} by linearizing \eqref{E:power_compact} around the nominal operating point $x_e$ satisfying $F(x_e)\!+\!BP_{in,0}\!=\!0$\,, where $P_{in,0}$ is the nominal power introduced in \eqref{E:power_nominal}. Generator power inputs are considered deterministic (Sec.\,\ref{Sec_model}), while loads (controllable and non-controllable) display both additive and multiplicative stochastic uncertainties as in \eqref{eq_stochastic_sys_compact}.

%------------------------------------
\subsection{Scenario 1: Full-State Measurements}
%------------------------------------
\begin{figure*}[thpb]
\centering
\subfigure[Pareto front of optimal solutions ($\gamma_1^*\!=\!0.06$)]{
\includegraphics[width=0.3\linewidth]{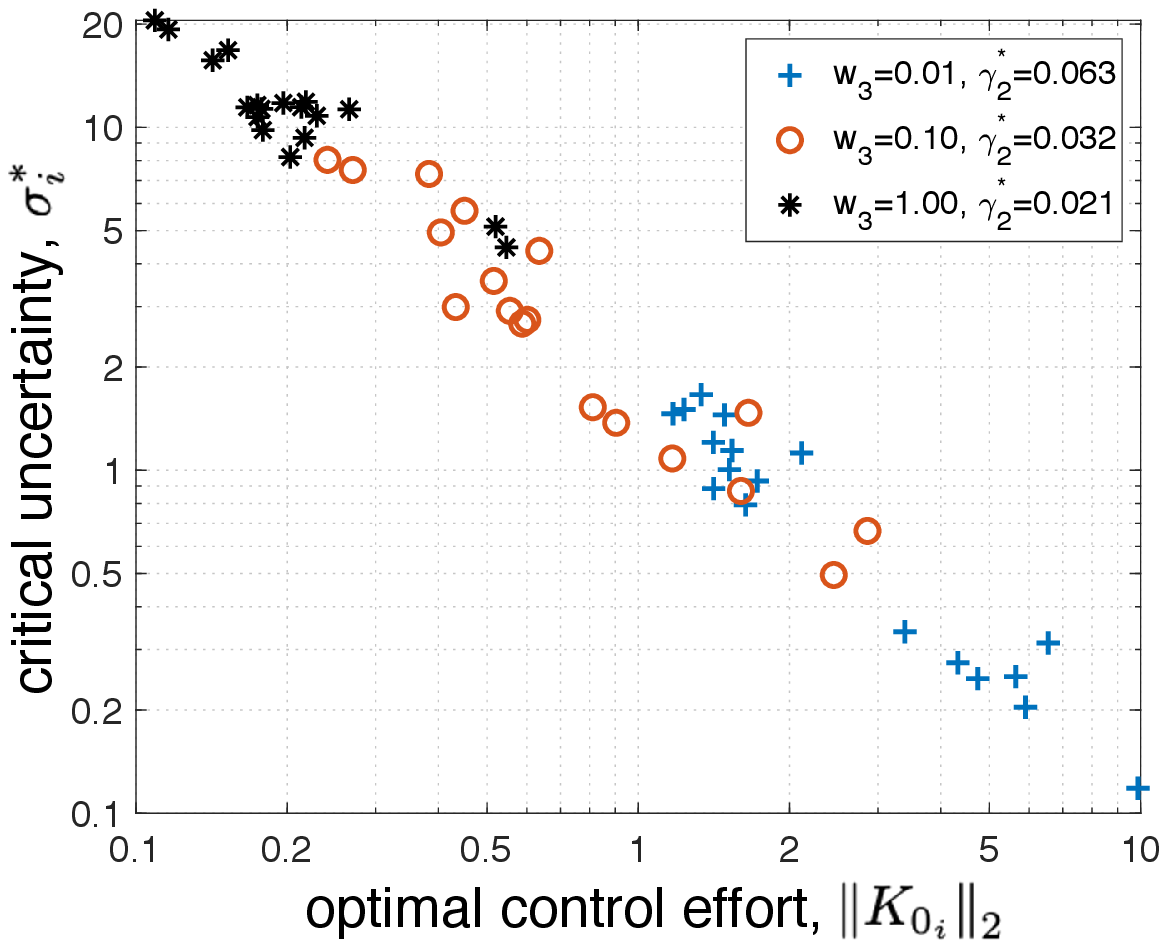}\label{fig:pareto}
}
% \hspace{-0.1in}
\subfigure[Control gain structure]{
\includegraphics[width=0.3\linewidth]{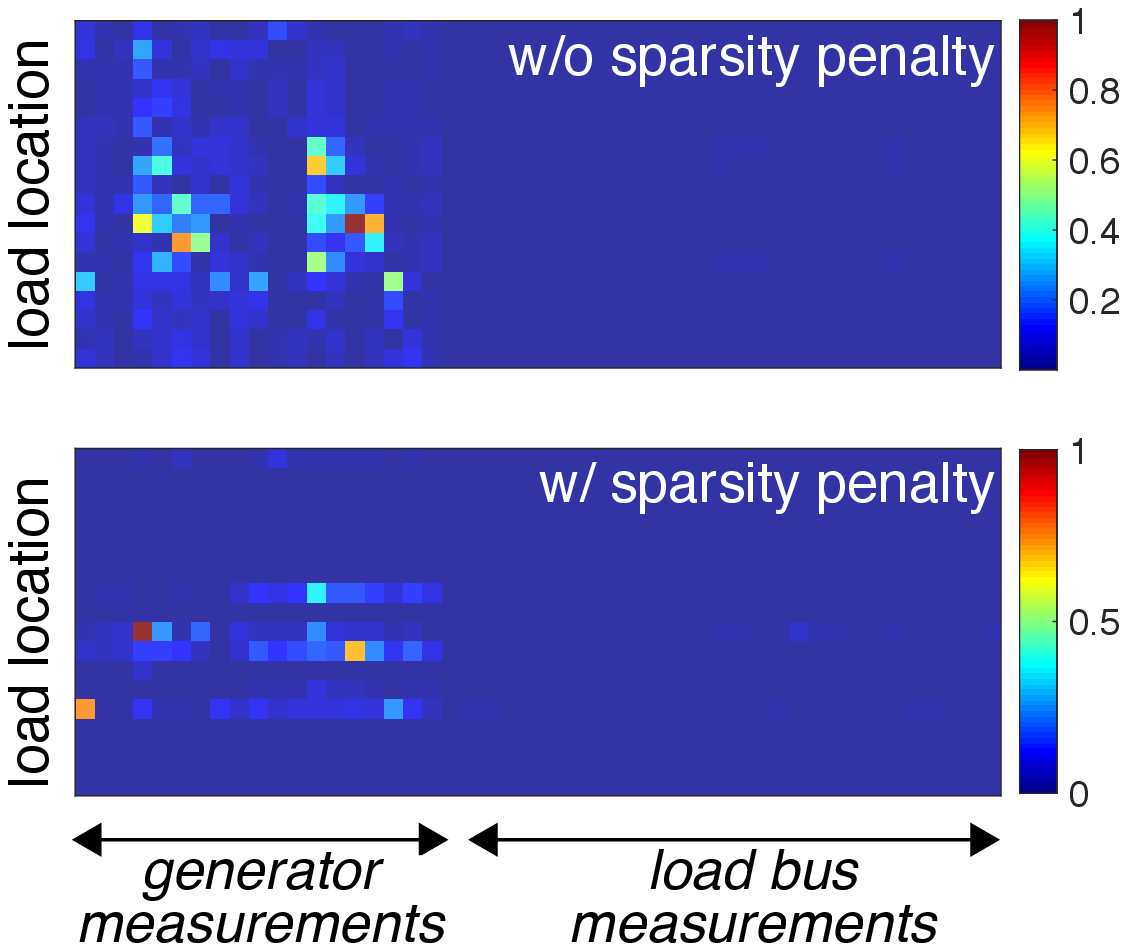}\label{fig:gain}
}
% \hspace{-0.1in}
\subfigure[Impact of penalties on sparsity ($\gamma_1^*\!=\!0.06$)]{
\includegraphics[width=0.3\linewidth]{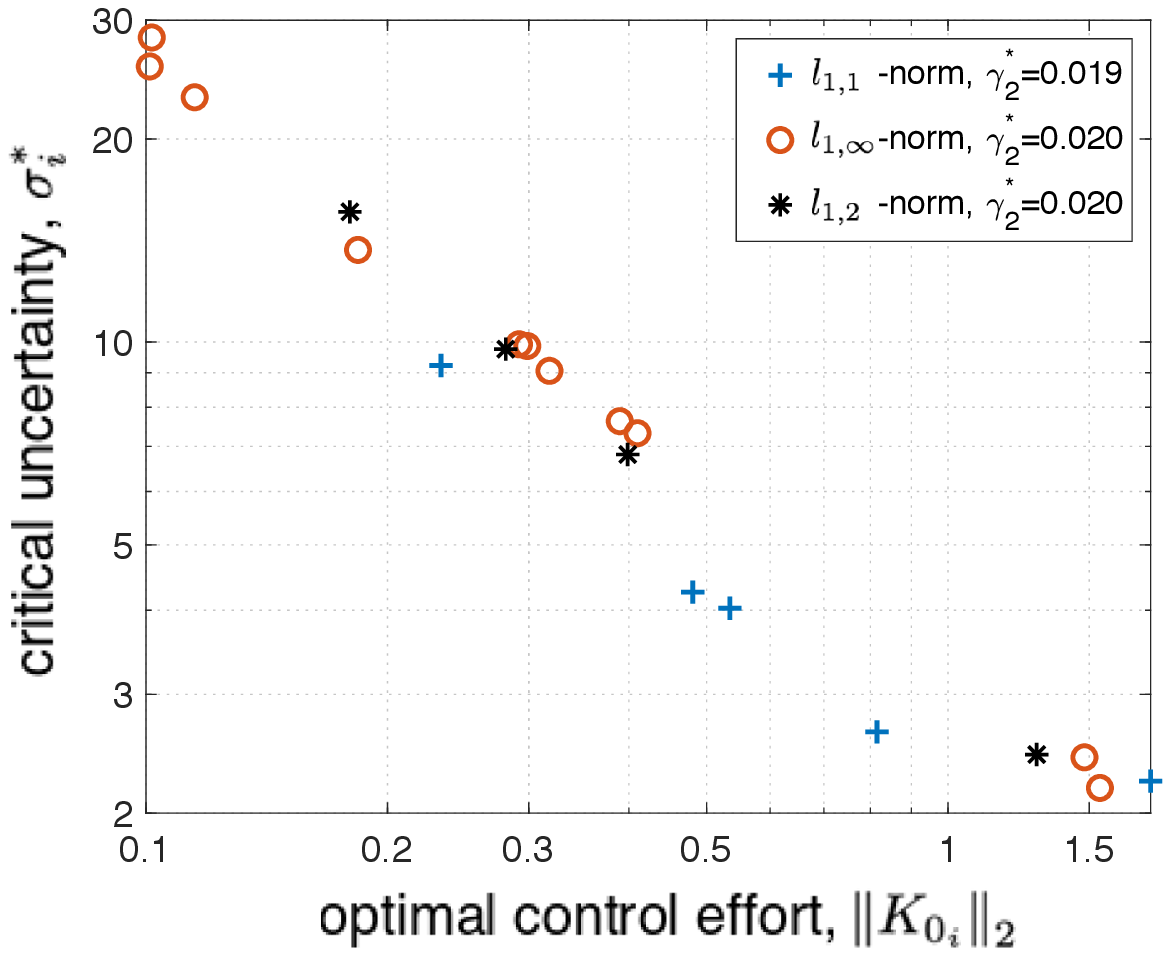}\label{fig:norm}
}
\caption[]{Analysis of the optimal solutions and the impact of various choices for optimization parameters: (a) shows the relations between optimal control effort ($\|K_{0_i}\|_2$), critical uncertainties ($\sigma_i^*$), and the transient performance measure ($\gamma_2^*$), as we vary the weight $w_3$, while keeping $w_1\!=\!w_2\!=\!100$; (b) shows an example of the controller gain matrix ($K_0$) structures with (\textit{bottom}) and without (\textit{top}) the sparsity penalty; (c) shows the impact of different choices for the $l_{1,p}$ norms of $\widetilde{M}$, defined in \eqref{eq:M_norm}, as a sparsity-promoting additional penalty term, while choosing $\left(w_1,w_2,w_3\right)\!=\!(100,100,0.1)$.}
\label{fig:performance}
\end{figure*}
\begin{figure}
    \centering
    \includegraphics[width=0.75\linewidth]{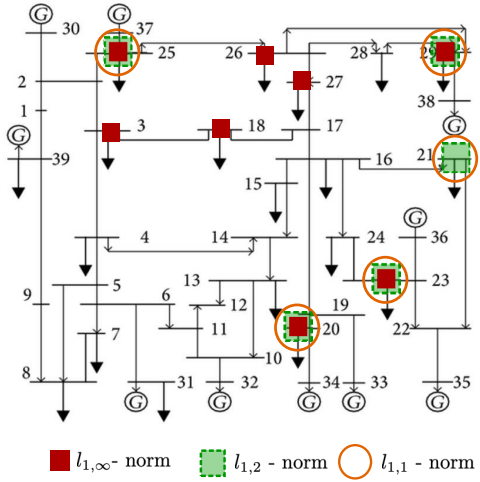}
    \caption{Illustration of the optimal sparse control locations, for three different choices of $l_{1,p}$-norms shown in Fig.\,\ref{fig:norm}. Both $l_{1,1}$ and $l_{1,2}$ norms pick the same load bus locations for control, but the load locations using the $l_{1,\infty}$-norm are somewhat different. All buses are equipped with sensors.}
    \label{fig:loc}
\end{figure}

\begin{figure}
    \centering
    \includegraphics[width=0.95\linewidth]{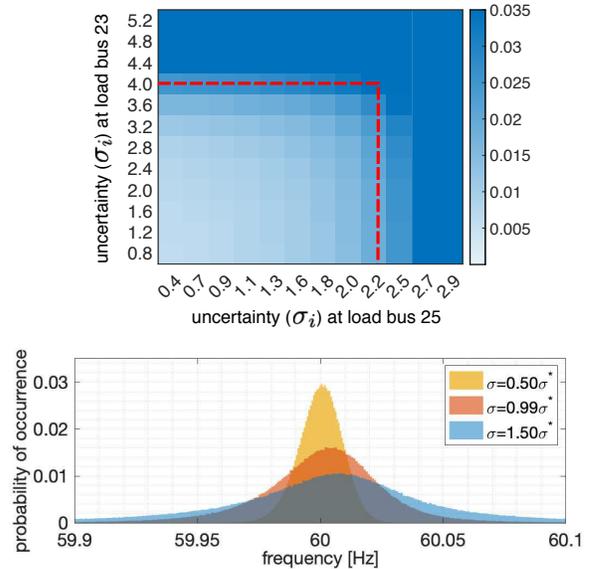}
    \caption{(Top) Variance of the closed-loop system frequency with varying uncertainty levels in controllable loads at two locations: bus 23 and bus 25 with respective critical uncertainty ($\sigma_i^*$) values of 4.03 and 2.23. (Bottom) System-wide frequency distributions at varying levels of uncertainties (relative to their critical values) in the controllable loads .}\vspace{-0.1in}
    \label{fig:sigma_1_2}
\end{figure}

\textbf{Analyzing Trade-off and Sparsity:} First we consider the scenario when each bus has a synchrophasor giving access to full-state measurements. Fig.\,\ref{fig:performance} illustrates the impact of various optimization choices, e.g. selection of the weights $w_{1,2,3}$ and inclusion/exclusion of the sparsity penalty term, on the optimal control effort ($\|K_{0_i}\|_2$), the critical uncertainties ($\sigma_i^*$), and the transient performance measure ($\gamma_2^*$). In particular, Fig.\,\ref{fig:pareto} shows the trade off between the optimal control effort ($\|K_{0_i}\|_2$) and the critical uncertainties ($\sigma_i^*$) at each load bus (plotted in the \textit{log-log} scale). Recall that this fundamental trade off is also suggested by the relation \eqref{eq:pareto}, especially since $\gamma_1^*\!=\!0.06$ remains same in the three cases. Moreover, as we increase the weight $w_3$ from 0.01 to 1, while keeping $w_1\!=\!w_2\!=\!100$, we observe that the transient performance measure ($\gamma_2^*$) deteriorates (from 0.063 to 0.021), in addition to lower control efforts and higher critical uncertainty levels. Fig.\,\ref{fig:gain} illustrates the impact of the addition of a sparsity-promoting penalty term, in the form of the norms ($\|\widetilde{M}\|_{1,p}$), on the structure of the optimal control gain matrix ($K_0$). Clearly, addition of the penalty term leads to a sparse control gain matrix. {As the rows of the controller gain matrix correspond to the controllable load location, the sparse controller provides frequency response services by engaging few controllable loads as opposed to the controller without enforcing the sparsity constraints. This shows the trade-offs between the sparsity and engaging the controllable resources which results in the reduction in cost of deployment of the sparse control for real-time implementation.} It is interesting to note that, since the transient dynamics is mostly dominated by the generator states, the control gains associated with the load bus measurements are relatively much lower compared to the generator bus measurements. Finally, Fig.\,\ref{fig:norm} shows the impact of the various choices of $p\!\in\!\lbrace 1,2,\infty\rbrace$ in the $l_{1,p}$ norm of $\widetilde{M}$ on the optimal solutions ($\gamma_1^*\!=\!0.06$ remains unchanged). The different norms appear to result in very close $\gamma_2^*$ values, with very little to differentiate in terms of the control effort and critical uncertainty values. The $l_{1,1}$-norm, however, seems to produce a sparse optimal control solution in which the control effort is relatively evenly distributed over the (selected few) control nodes. Fig.\,\ref{fig:loc} shows the optimal control locations for the different norms (as per Fig.\,\ref{fig:norm}), with an agreement between the $l_{1,1}$ and $l_{1,2}$-norm solutions, while the $l_{1,\infty}$-norm solution displays some differences.

\textbf{Overall Performance Evaluation:} Next we implement the optimal sparse controller (obtained using the $l_{1,1}$-norm) in closed-loop transient simulations of the IEEE 39-bus network. Controller performance is measured with respect to the variance in the observed grid frequencies under load uncertainties {via stochastic time-domain simulations. These simulations serve as a validation for the analytical findings such as the critical uncertainty identified from Eq.} \eqref{eq:critical_unc_full_obs} {and Eq.} \eqref{eq:sigma}. 
In particular, we evaluate the performance for various levels of load uncertainties, relative to the associated critical uncertainty values. Fig.\,\ref{fig:sigma_1_2} (top) shows the impact of varying uncertainties at two load locations: bus 23 and bus 25 with respective $\sigma_i^*$ values of 4.03 and 2.23. Closed-loop stochastic stability performance deteriorates (with increased frequency variance) as the uncertainties in the (controllable) loads increase. {The red-dotted line in Fig.}\,\ref{fig:sigma_1_2} {(top) shows the critical uncertainty obtained analytically from Eq.} \eqref{eq:critical_unc_full_obs}. 
Figure\,\ref{fig:sigma_1_2} (bottom) depicts a similar story via the frequency distributions, for three different uncertainty levels in the controllable loads across the system. As the uncertainty is increased, the distribution of frequency becomes heavy tailed indicating the frequent excursions of frequency away from the nominal acceptable range of operation (in line with the observations in \cite{pushpak2018stochastic}).

\begin{figure*}[thpb]
\centering
\includegraphics[width=0.95\linewidth]{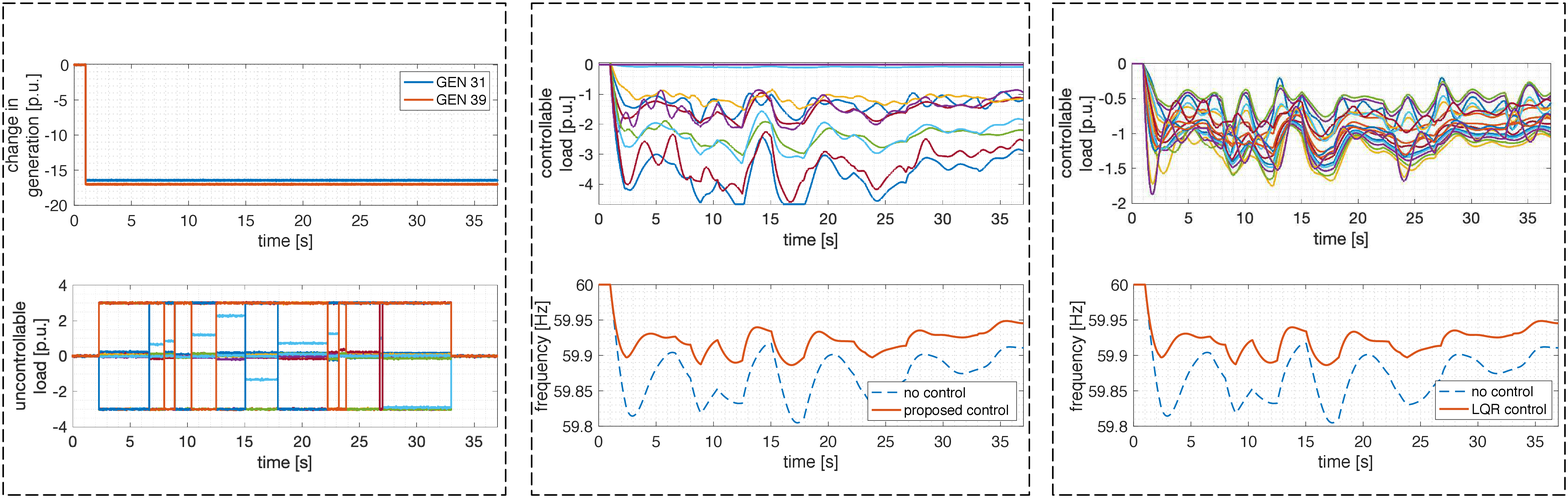}
\caption[]{{Time-domain plots to illustrate the performance of the proposed controller against a contingency scenario of large loss of generation resulting in an under-frequency event. \textbf{Left} plots show the loss of the generation at buses 31 and 39, along with transient disturbances in the uncontrollable loads. \textbf{Middle} plots show the controllable load power values, along with the improved frequency profile, when the \textit{proposed sparse} controller is applied, designed using $(w_1,w_2,w_3)=(2000,2000,0.2)$. \textbf{Right} plots show the action of a conventional LQR controller tuned to provide similar frequency response as the proposed sparse controller. Note the sparsity in controllable load power in the case of proposed controller, as compared to the LQR controller.}}
\label{fig:under}
\end{figure*}

\begin{figure*}[thpb]
\centering
\includegraphics[width=0.95\linewidth]{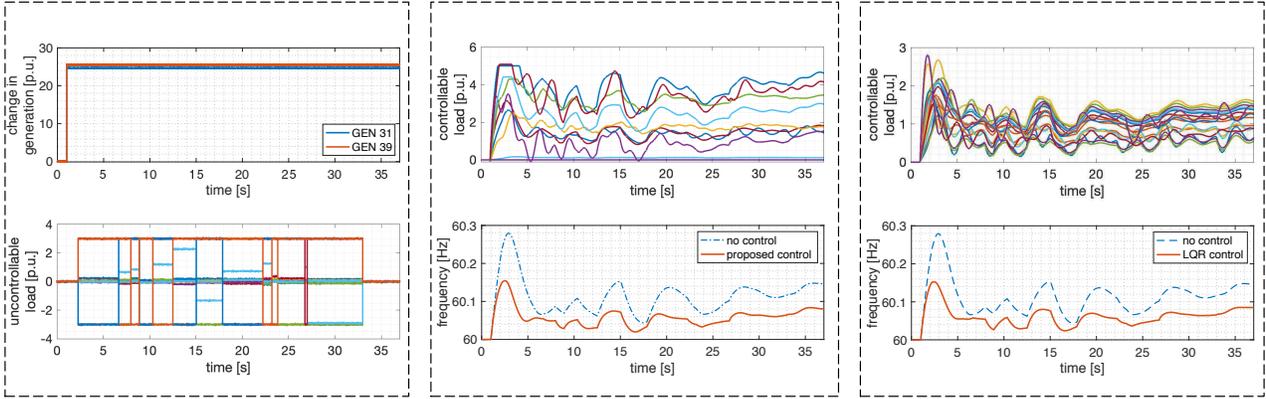}
\caption[]{{Time-domain plots to illustrate the performance of the proposed controller against a contingency scenario of sudden increase in generation (e.g., solar) resulting in an over-frequency event. \textbf{Left} plots show the increase in generation at buses 31 and 39, along with transient disturbances in the uncontrollable loads. \textbf{Middle} plots show the controllable load power values, along with the improved frequency profile, when the \textit{proposed sparse} controller is applied, designed using $(w_1,w_2,w_3)=(2000,2000,0.2)$. \textbf{Right} plots show the action of a conventional LQR controller tuned to provide similar frequency response as the proposed sparse controller. Note the sparsity in controllable load power in the case of proposed controller, as compared to the LQR controller.}}
\label{fig:over}
\end{figure*}

\begin{figure}[thpb]
\centering
\includegraphics[width=0.7\linewidth]{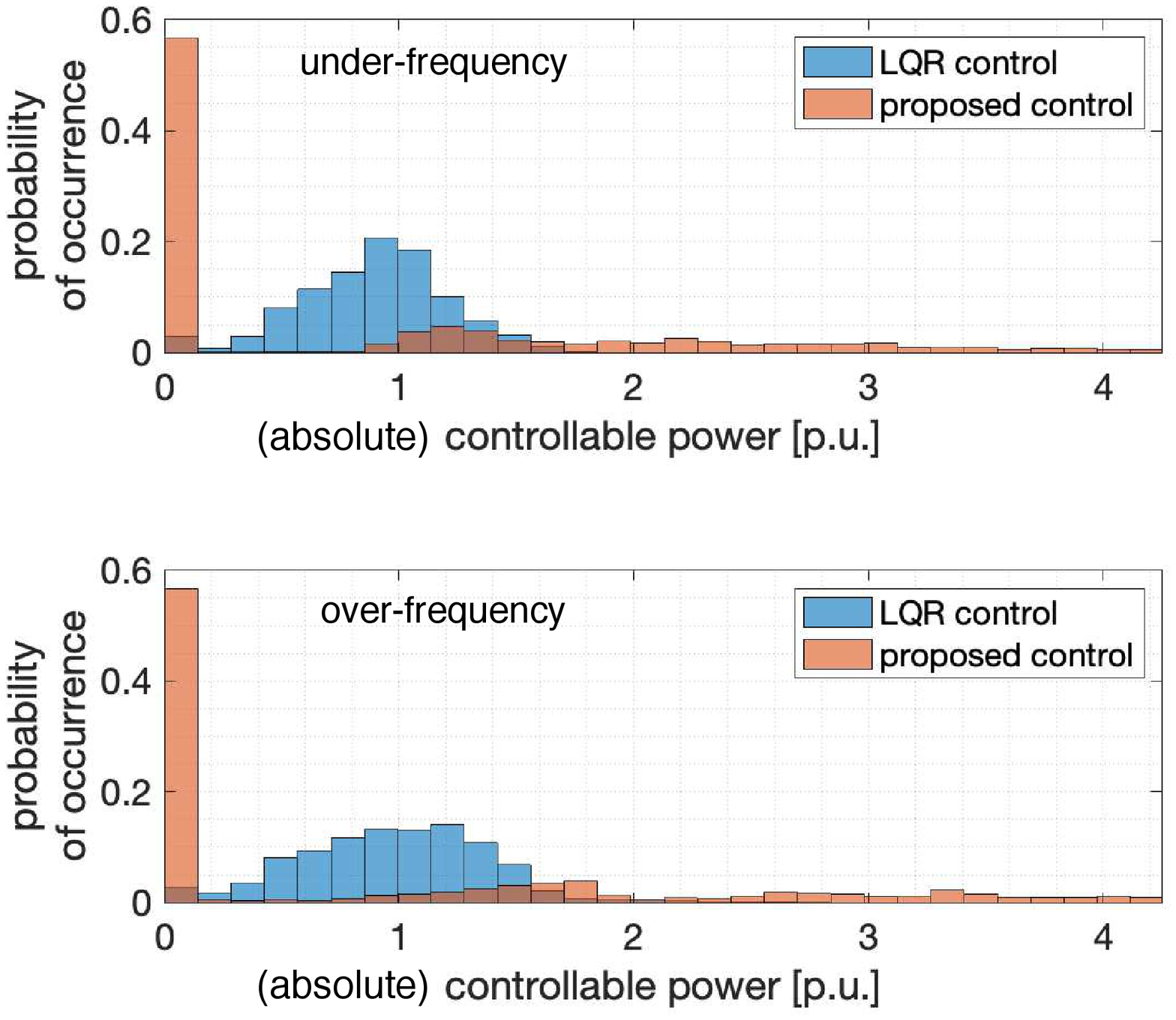}
\caption[]{{Comparison of the probability distributions of the absolute magnitude of the controllable load changes (in p.u.) for both the proposed and the LQR controllers. In both the under- and over-frequency contingency cases, the sparse controller engages only about 45\% of the load locations (8 out of a total of 18}, with about 55\% of controllable power set to 0\,. The LQR is shown to distribute the controllable power evenly across all the load locations.}
\label{fig:hist}
\end{figure}

{\textbf{Performance under Contingencies:}  Figs.\,}\ref{fig:over} and \ref{fig:under} {demonstrate via time-domain plots the closed-loop control performance of the proposed controller and a conventional linear quadratic regulator (LQR) control towards stochastic stabilization of the system under two contingencies -- specifically, a significant loss of generation at buses 31 and 39 resulting in large frequency drop (Fig.\,}\ref{fig:under}{), and a sudden large increase in generation at buses 31 and 39 (e.g., due to renewable fluctuations) resulting in over-frequency event (Fig.\,}\ref{fig:over}{). In both the scenarios, frequent disturbances are simulated in the non-controllable loads, along with the presence of uncertainties. The plots in the middle column demonstrate the performance of the proposed controller in providing \textit{primary frequency response} by reducing the controllable load power to improve the frequency profile. In Fig.\,}\ref{fig:under}{, the proposed sparse controller is shown to successfully increase frequency nadir from around 59.8\,Hz (without any control) to 59.9\,Hz (with control), by only engaging a few of the controllable load locations (e.g., only 8 of the 18 load locations). In Fig.\,}\ref{fig:over} {the highest frequency value is seen to decrease from around 60.3\,Hz (without any control) to around 60.15\,Hz (with control). In order to compare the performance of the proposed sparse controller, we design and implement a conventional LQR control, tuned appropriately to match the closed-loop frequency profile of the sparse controller. We use the controllable load power values between the two controllers to illustrate the impact of sparsity. In particular, the controllable load power plots in Figs.\,}\ref{fig:under} and \ref{fig:over} {reveal that while the conventional LQR engages all of the 18 available controllable loads across the network, the sparse controller only engages 8 of those. This is further exemplified in Fig.\,}\ref{fig:hist} {in which we compare the probability distributions of the absolute magnitude of the controllable load power values for the LQR controller and the proposed sparse controller. As expected, the plots reveal the relatively high percentage (around 55\%) of the loads that are never engaged by the sparse controller, while the LQR engages all the loads.}

{Finally, since the proposed power system model in Section} \ref{Sec_model} {and} \eqref{E:power_dyn} {consists of aggregate model for the controllable as well as the uncontrollable loads, the proposed sparse control synthesis is oblivious to the type of end-use devices and yields a controller that sets the value of the aggregated controllable power. The controllable loads aggregation may consist of numerous flexible devices, including inverter-interfaced devices such as a solar PV or a battery. The design of a control scheme that real-time engages the devices in the aggregation to provide primary frequency response has been extensively studied in prior efforts} \cite{moya2014hierarchical,lian2016hierarchical,nandanoori2018prioritized}{, and hence not elaborated further in this work. However, the system-wide synthesis of the proposed sparse controller illustrates the promise of the flexible loads in providing frequency response services under uncertainties. }

%------------------------------------
\subsection{Scenario 2: Partial-State Measurements}
%------------------------------------
We now consider a scenario when the synchrophasors are placed only on the generator buses, as typical of realistic power systems. The result of the two-stage optimal sparse control synthesis method described in Sec.\,\ref{sec:control_partial} is shown in Fig.\,\ref{fig:loc_partial}. The Fig.\,\ref{fig:loc_partial} top plot displays the optimal control solution with the generator-only measurements, in comparison with the full-state measurements (using the $l_{1,1}$-norm). The closeness of the two solutions can be intuitively explained based on the relative sparsity of the columns of control gain matrix (shown in Fig.\,\ref{fig:gain}) associated with the load bus measurements. The Fig.\,\ref{fig:loc_partial} bottom plot illustrates the sensing and control locations.
\begin{figure}
    \centering
    \includegraphics[width=0.75\linewidth]{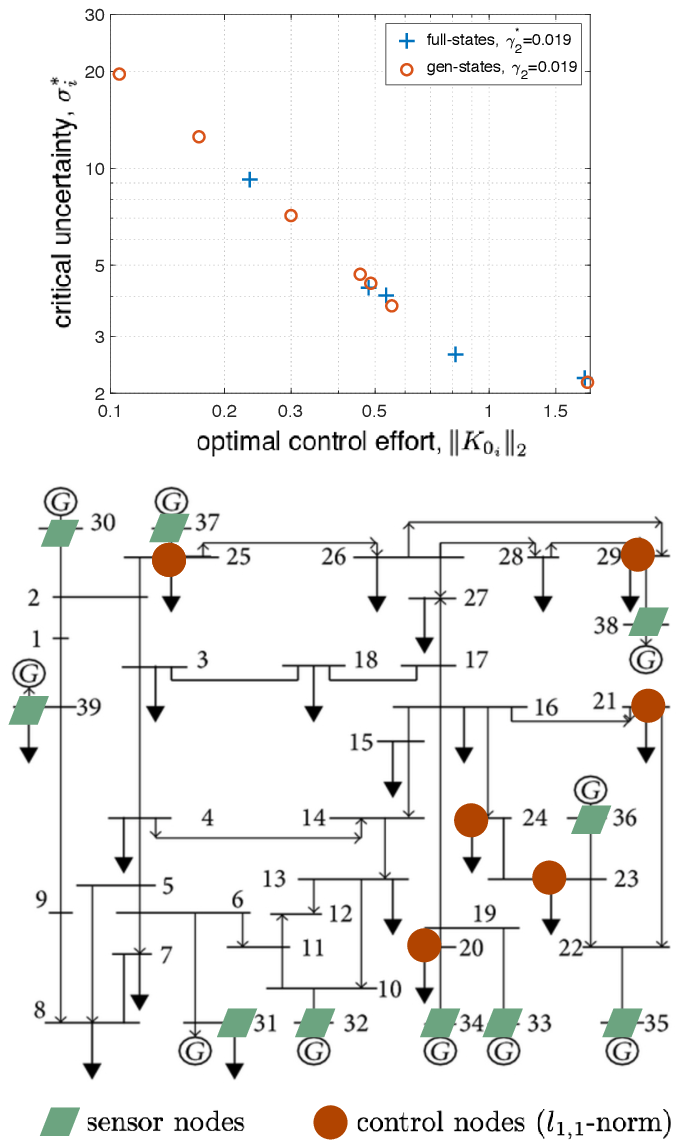}
    \caption{(Top) Two-stage optimal sparse control solution using the generator-state measurements alone is shown in comparison with the full-state measurement solution. (Bottom) Illustration of the sensor and control locations.}
    \label{fig:loc_partial}
\end{figure}

%------------------------------------
\subsection{Additional Remarks}
%------------------------------------

{\textbf{Computational Feasibility:} For the IEEE 39 bus system studies given in Fig.} \ref{fig:performance}, {there are a total of 48 states and 18 controllable load locations. The resulting sparse control synthesis (under full observation) problem leads to semi-definite programs (SDPs) with matrix variables for sizes} $48\times 48$. {Solved via CVX} \cite{grant2008cvx} {using Matlab\textregistered 2020b, the control synthesis problem takes around 43 seconds on a computer with the following configuration: 2.6 GHz Intel core i7 processor with 16GB of RAM. Furthermore, since the controller is designed offline and just uses the output measurements, the control execution is nearly instantaneous.
Computational feasibility of SDPs, such as the ones in} \eqref{eq:controller_synthesis}, \eqref{eq:controller_synthesis_sparse}, and \eqref{eq:controller_synthesis_stage1}-\eqref{eq:controller_synthesis_stage2}, {for large-scale problems is an active area of research. Recent results} \cite{majumdar2020recent,yurtsever2021scalable} {show promise for solving large-scale SDPs through economic storage management and cost effective arithmetic operations. In particular, the work in} \cite{yurtsever2021scalable} {reports} $\mathcal{O}(n)$ {performance (i.e., linear scaling) in terms of storage requirement and computational cost for SDPs of size} $n$ {across a wide range SDP solvers, including SDPT3} \cite{sdpt3} {and} SeDuMi \cite{sedumi} {available with CVX. 
}

\vspace{0.5 cm}

\noindent\textbf{Key Observation:} The numerical studies demonstrate that the LMI-based approach provides a systematic way of designing (and placing) optimal and sparse demand-side controllers that can guarantee desirable transient frequency stabilization under uncertainties in the controllable (and non-controllable) loads. The control synthesis method extends to the realistic scenario of partial-state measurements, such as when only the generator buses are equipped with synchrophasors. Moreover, a quantifiable measure of maximal allowable (critical) uncertainties at any load location is obtained, and the closed-loop frequency response performance demonstrably degrades at uncertainties beyond these critical levels.

%%%%%%%% OLD SECTIONS %%%%%%%%%%

% \input{old_sections/modeling.tex}
% \input{old_sections/controllable_loads.tex}
% \input{old_sections/robust_control_old.tex}
% \input{old_sections/simulation_old.tex}

%%%%%%%%%%%%%%%%%%%%%%%%%%%%%%%%%

\vspace{-0.1in}\section{Conclusion}
\label{Sec_conclusion}
This work considers the problem of achieving frequency response through demand-side control in the presence of uncertainty in the controllable as well as uncontrollable loads. An LMI-based demand-side frequency control synthesis problem is formulated for the power networks, that ensures stochastic stability and, in a multi-objective setting, promotes sparsity of the controller, maximizes allowable uncertainties, and enhances closed-loop transient performance.
%
% An IEEE 39 bus system modeled with controllable loads is considered for demonstrating the proposed sparse control synthesis. 
Extensive numerical studies on the IEEE 39-bus system show the effect of penalizing for sparsity and different types of sparsity norms, impact of critical uncertainty on the system states which results in a heavy tailed distribution for frequencies. The Pareto front of optimal solutions between the optimal control efforts and critical uncertainties indicate that allowing more uncertainty in the controllable loads results in decreased closed-loop transient performance. 
% Study on the various sparsity norms show a similar closed-loop performance although the maximal tolerable uncertainties and the corresponding optimal control efforts vary. 
Closed-loop stochastic stability is demonstrated even with only the generator states as measurements. 
% A power network modeled with aggregate residential loads that can be leveraged as controllable loads to provide frequency response services is considered in this work. The dynamic behavior in the availability of controllable load power is modelled as stochastic uncertainty which appears multiplicative in the system dynamics. A robust controller that can tolerate maximal uncertainty at every load bus is formulated as an LMI-based optimization problem from the necessary and sufficient mean square exponential stability conditions. The proposed robust control design is illustrated on an IEEE $39$ bus system and critical load locations are identified. Further, the effect of these uncertainties on the system frequency and the performance of the robust controller are demonstrated using extensive time-domain simulations. 
% Finally, the robustness of the controller is put to test with significant load changes in the presence of stochastic uncertainty in the controllable loads. 
%
%
%
Future work will address the extension of this control synthesis approach to uncertain generating resources (e.g. solar and wind); as well as explore the co-design problem of placement of sensing and control nodes, via enforcing both column- and row-sparsity on the control gain matrices.

% --------------------
% \section*{Acknowledgment}

% --------------------

\ifCLASSOPTIONcaptionsoff
  \newpage
\fi

\bibliographystyle{IEEEtran}
\bibliography{Refs.bib}

\end{document}